\documentclass[aps,pra,twocolumn,notitlepage,superscriptaddress,showpacs,nofootinbib]{revtex4-1}
\usepackage{amsmath}
\usepackage{amssymb}
\usepackage{amsfonts}
\usepackage{enumerate}
\usepackage{mathrsfs}
\usepackage{graphicx}
\usepackage{epstopdf}
\usepackage{enumerate}
\usepackage{changepage}
\usepackage{bm}
\usepackage{multirow}
\usepackage{lipsum}
\usepackage{booktabs}

\newtheorem{theorem}{Theorem}
\newtheorem{definition}{Definition}

\newtheorem{lemma}{Lemma}
\newtheorem{proposition}{Proposition}
\newtheorem{conjecture}{Conjecture}
\newtheorem{example}{Example}
\newtheorem{corollary}{Corollary}

\def\bcj{\begin{conjecture}}
	\def\ecj{\end{conjecture}}
\def\bcr{\begin{corollary}}
	\def\ecr{\end{corollary}}
\def\bd{\begin{definition}}
	\def\ed{\end{definition}}
\def\bea{\begin{eqnarray}}
	\def\eea{\end{eqnarray}}
\def\bem{\begin{enumerate}}
	\def\eem{\end{enumerate}}
\def\bex{\begin{example}}
	\def\eex{\end{example}}
\def\bim{\begin{itemize}}
	\def\eim{\end{itemize}}
\def\bl{\begin{lemma}}
	\def\el{\end{lemma}}
\def\bma{\begin{bmatrix}}
	\def\ema{\end{bmatrix}}
\def\bpf{\begin{proof}}
	\def\epf{\end{proof}}
\def\bpp{\begin{proposition}}
	\def\epp{\end{proposition}}
\def\bqu{\begin{question}}
	\def\equ{\end{question}}
\def\br{\begin{remark}}
	\def\er{\end{remark}}
\def\bt{\begin{theorem}}
	\def\et{\end{theorem}}

\def\squareforqed{\hbox{\rlap{$\sqcap$}$\sqcup$}}
\def\qed{\ifmmode\squareforqed\else{\unskip\nobreak\hfil
		\penalty50\hskip1em\null\nobreak\hfil\squareforqed
		\parfillskip=0pt\finalhyphendemerits=0\endgraf}\fi}
\def\endenv{\ifmmode\;\else{\unskip\nobreak\hfil
		\penalty50\hskip1em\null\nobreak\hfil\;
		\parfillskip=0pt\finalhyphendemerits=0\endgraf}\fi}
\newenvironment{proof}{\noindent \textbf{{Proof.~} }}{\qed}
\def\Dbar{\leavevmode\lower.6ex\hbox to 0pt
	{\hskip-.23ex\accent"16\hss}D}
\makeatletter
\def\url@leostyle{%
	\@ifundefined{selectfont}{\def\UrlFont{\sf}}{\def\UrlFont{\small\ttfamily}}}
\makeatother
\def\bcj{\begin{conjecture}}
	\def\ecj{\end{conjecture}}
\def\bcr{\begin{corollary}}
	\def\ecr{\end{corollary}}
\def\bd{\begin{definition}}
	\def\ed{\end{definition}}
\def\bea{\begin{eqnarray}}
	\def\eea{\end{eqnarray}}
\def\bem{\begin{enumerate}}
	\def\eem{\end{enumerate}}
\def\bex{\begin{example}}
	\def\eex{\end{example}}
\def\bim{\begin{itemize}}
	\def\eim{\end{itemize}}
\def\bl{\begin{lemma}}
	\def\el{\end{lemma}}
\def\bpf{\begin{proof}}
	\def\epf{\end{proof}}
\def\bpp{\begin{proposition}}
	\def\epp{\end{proposition}}
\def\bqu{\begin{question}}
	\def\equ{\end{question}}
\def\br{\begin{remark}}
	\def\er{\end{remark}}
\def\bt{\begin{theorem}}
	\def\et{\end{theorem}}

\def\btb{\begin{tabular}}
	\def\etb{\end{tabular}}

	\newcommand{\nc}{\newcommand}
	
	\nc{\bbA}{\mathbb{A}} \nc{\bbB}{\mathbb{B}} \nc{\bbC}{\mathbb{C}}
	\nc{\bbD}{\mathbb{D}} \nc{\bbE}{\mathbb{E}} \nc{\bbF}{\mathbb{F}}
	\nc{\bbG}{\mathbb{G}} \nc{\bbH}{\mathbb{H}} \nc{\bbI}{\mathbb{I}}
	\nc{\bbJ}{\mathbb{J}} \nc{\bbK}{\mathbb{K}} \nc{\bbL}{\mathbb{L}}
	\nc{\bbM}{\mathbb{M}} \nc{\bbN}{\mathbb{N}} \nc{\bbO}{\mathbb{O}}
	\nc{\bbP}{\mathbb{P}} \nc{\bbQ}{\mathbb{Q}} \nc{\bbR}{\mathbb{R}}
	\nc{\bbS}{\mathbb{S}} \nc{\bbT}{\mathbb{T}} \nc{\bbU}{\mathbb{U}}
	\nc{\bbV}{\mathbb{V}} \nc{\bbW}{\mathbb{W}} \nc{\bbX}{\mathbb{X}}
	\nc{\bbZ}{\mathbb{Z}}
	
	\nc{\bA}{{\bf A}} \nc{\bB}{{\bf B}} \nc{\bC}{{\bf C}}
	\nc{\bD}{{\bf D}} \nc{\bE}{{\bf E}} \nc{\bF}{{\bf F}}
	\nc{\bG}{{\bf G}} \nc{\bH}{{\bf H}} \nc{\bI}{{\bf I}}
	\nc{\bJ}{{\bf J}} \nc{\bK}{{\bf K}} \nc{\bL}{{\bf L}}
	\nc{\bM}{{\bf M}} \nc{\bN}{{\bf N}} \nc{\bO}{{\bf O}}
	\nc{\bP}{{\bf P}} \nc{\bQ}{{\bf Q}} \nc{\bR}{{\bf R}}
	\nc{\bS}{{\bf S}} \nc{\bT}{{\bf T}} \nc{\bU}{{\bf U}}
	\nc{\bV}{{\bf V}} \nc{\bW}{{\bf W}} \nc{\bX}{{\bf X}}
	\nc{\ba}{{\bf a}} \nc{\be}{{\bf e}} \nc{\bu}{{\bf u}}
	\nc{\brr}{{\bf r}}
	
	\nc{\cA}{{\cal A}} \nc{\cB}{{\cal B}} \nc{\cC}{{\cal C}}
	\nc{\cD}{{\cal D}} \nc{\cE}{{\cal E}} \nc{\cF}{{\cal F}}
	\nc{\cG}{{\cal G}} \nc{\cH}{{\cal H}} \nc{\cI}{{\cal I}}
	\nc{\cJ}{{\cal J}} \nc{\cK}{{\cal K}} \nc{\cL}{{\cal L}}
	\nc{\cM}{{\cal M}} \nc{\cN}{{\cal N}} \nc{\cO}{{\cal O}}
	\nc{\cP}{{\cal P}} \nc{\cQ}{{\cal Q}} \nc{\cR}{{\cal R}}
	\nc{\cS}{{\cal S}} \nc{\cT}{{\cal T}} \nc{\cU}{{\cal U}}
	\nc{\cV}{{\cal V}} \nc{\cW}{{\cal W}} \nc{\cX}{{\cal X}}
	\nc{\cZ}{{\cal Z}}
	
	\nc{\hA}{{\hat{A}}} \nc{\hB}{{\hat{B}}} \nc{\hC}{{\hat{C}}}
	\nc{\hD}{{\hat{D}}} \nc{\hE}{{\hat{E}}} \nc{\hF}{{\hat{F}}}
	\nc{\hG}{{\hat{G}}} \nc{\hH}{{\hat{H}}} \nc{\hI}{{\hat{I}}}
	\nc{\hJ}{{\hat{J}}} \nc{\hK}{{\hat{K}}} \nc{\hL}{{\hat{L}}}
	\nc{\hM}{{\hat{M}}} \nc{\hN}{{\hat{N}}} \nc{\hO}{{\hat{O}}}
	\nc{\hP}{{\hat{P}}} \nc{\hR}{{\hat{R}}} \nc{\hS}{{\hat{S}}}
	\nc{\hT}{{\hat{T}}} \nc{\hU}{{\hat{U}}} \nc{\hV}{{\hat{V}}}
	\nc{\hW}{{\hat{W}}} \nc{\hX}{{\hat{X}}} \nc{\hZ}{{\hat{Z}}}
	
	\nc{\hn}{{\hat{n}}}
	
	\def\dim{\mathop{\rm Dim}}

	\def\rank{\mathop{\rm rank}}
	
	\def\tr{\mathop{\rm Tr}}

	\newcommand{\ket}[1]{|#1\rangle}
	
	\newcommand{\ketbra}[2]{|#1\rangle\!\langle#2|}
	\newcommand{\braket}[2]{\langle#1|#2\rangle}

	\usepackage[
	colorlinks,
	linkcolor = blue,
	citecolor = blue,
	urlcolor = blue]{hyperref}
	\def \qed {\hfill \vrule height7pt width 7pt depth 0pt}
	
	\newcounter{lastnote}

\begin{document}
		\title{Unextendible and uncompletable product bases in every bipartition}
	\author{Fei Shi}
\email[]{shifei@mail.ustc.edu.cn}
\affiliation{School of Cyber Security,
	University of Science and Technology of China, Hefei, 230026,  China}
\affiliation{Department of Computer Science,
The University of Hong Kong, Pokfulam Road, Hong Kong,  China}

\author{Mao-Sheng Li}
	\email{li.maosheng.math@gmail.com}
	\affiliation{ School of Mathematics,
		South China University of Technology, Guangzhou
		510641,  China}


\author{Xiande Zhang}
\email[]{drzhangx@ustc.edu.cn}
\affiliation{School of Mathematical Sciences,
	University of Science and Technology of China, Hefei, 230026,  China}
		
\author{Qi Zhao}
\email[]{zhaoqi@cs.hku.hk}
\affiliation{Department of Computer Science,
The University of Hong Kong, Pokfulam Road, Hong Kong,  China}


		\begin{abstract}
			
		    Unextendible product basis is an important object in quantum information theory and features a broad spectrum of applications, ranging bound entangled states, quantum nonlocality without entanglement, and Bell inequalities with no quantum violation. A generalized concept called uncompletable product basis also attracts much attention. In this paper,  we find some unextendible product bases that are uncompletable product bases in every bipartition, which answers a 19 year-old open question proposed by DiVincenzo \emph{et al.} [\href{https://link.springer.com/article/10.1007/s00220-003-0877-6}{Commun. Math. Phys. \textbf{238}, 379 (2003)}]. As a consequence, we connect such unextendible product bases to local hiding of information and give a sufficient condition for the existence of an unextendible product basis, that is still an unextendible product basis in every bipartition. Our results advance the understanding of the geometry of unextendible product bases.

		\end{abstract}

		\maketitle
		\section{Introduction}\label{sec:int}
An unextendible product basis (UPB) in a multipartite quantum system is an incomplete orthogonal
product basis whose complementary subspace contains no product state \cite{bennett1999unextendible}.  UPBs have a lot of applications in quantum information. The mixed state that is proportional to the projector on the complementary subspace of any UPB is a positive-partial-transpose (PPT) entangled state.  PPT entangled states represent
the so-called bound entangled states from which no pure
entanglement can be distilled under local operations and classical communication (LOCC)  \cite{bennett1999unextendible}. Quantum nonlocaltiy is another important application. UPBs can not be perfectly distinguished under local positive operator-valued measures (POVMs)
and classical communication \cite{bennett1999unextendible}, which shows the phenomenon of quantum nonlocality without entanglement \cite{bennett1999quantum}. For perfect discrimination of UPBs, one can use entanglement resources \cite{cohen2008understanding,zhang2020locally}. Some UPBs are locally irreducible in every biparition, and showed the phenomenon of strong quantum nonlocality without entanglement \cite{Halder2019Strong,shi2021strong,shi2022strongly}. UPBs also can be used to show more nonlocality with less purity \cite{bandyopadhyay2011more}. Bell nonlocality is from  Bell inequalities, and UPBs were connected to Bell inequalities  with no quantum violation \cite{augusiak2011bell,Augusiak2012tight}.

In 2003, DiVincenzo \emph{et al.} generalized the concept of UPBs \cite{divincenzo2003unextendible}. An uncompletable product basis (UCPB) in a multipartite quantum system is an incomplete orthogonal product basis, which can not be extended to a complete orthogonal product basis \cite{divincenzo2003unextendible}. An incomplete orthogonal product basis is a strongly uncompletable product basis (SUCPB), if it is a UCPB in any locally extended Hilbert space \cite{divincenzo2003unextendible}. Actually, the set of all UPBs is a proper subset of the set of all  SUCPBs, and the set of all SUCPBs is a proper subset of the set of all  UCPBs. See also Fig.~\ref{fig:inclusion} for the inclusion relation of these three sets.
It is known that UPBs and SUCPBs cannot be perfectly distinguished under local POVMs and classical communication, and UCPBs cannot be perfectly distinguished under local projective measurements and classical communication \cite{bennett1999unextendible,divincenzo2003unextendible}. In \cite{divincenzo2003unextendible}, DiVincenzo \emph{et al.} proposed an open question: whether there exists a UPB, which is a UCPB in every bipartition? This open question exists for 19 years because there are few constructions of UPBs in multipartite systems, and it is difficult to show UCPBs in bipartite systems.  Such UPBs can be used to understand the geometry of UPBs. There exists another famous open question for UPBs \cite{demianowicz2018unextendible}: can we find a UPB, which is still a UPB in every bipartition? Such UPBs cannot be perfectly distinguished under local POVMs and classical communication in every bipartiton \cite{bennett1999unextendible}, and can be used to construct
genuinely entangled subspaces \cite{demianowicz2018unextendible}. Recently, Demianowice showed that  such UPBs with the minimum size do not exist \cite{demianowicz2022genuinely}. However, the existence of such UPBs is still unknown.

In this work, we address the 19 year-old open question in \cite{divincenzo2003unextendible}, by presenting a UPB with a stronger property, which is an SUCPB in every bipartition. We also show that such UPBs can be used for local hiding of information. Tile structures in bipartite systems provide an efficient method for constructing bipartite UPBs \cite{divincenzo2003unextendible,shi2020unextendible}. We generalize the tile structures to multipartite systems, and give a sufficient condition for the existence a UPB that is still a UPB in every bipartition. This sufficient condition is intuitive, and one can search such UPBs through computer under this condition.

The rest of this paper is organized as follows. In Sec.~\ref{sec:pre}, we introduce the concepts of UPBs, UCPBs, and SUCPBs. Next, in Sec.~\ref{sec:sucpb}, we find a UPB that is an SUCPBs in every bipartition for arbitrary three-, and four-partite system. In Sec.~\ref{sec:gupb}, we give  a sufficient condition for the existence of a UPB that is still a UPB in every bipartition. Finally, we conclude in Sec.~\ref{sec:con}.

		\section{Preliminaries}\label{sec:pre}
		In this paper, we do not normalize product states for simplicity. We denote $\bbZ_n:=\{0,1,\ldots,n-1\}$ and $w_n:=e^{\frac{2\pi i}{n}}$. For a matrix $M$, let $\text{sum}(M)$ be the sum of all elements.  Assume $\{\ket{i}_A\}_{i\in\bbZ_m}$ and $\{\ket{j}_B\}_{j\in\bbZ_n}$ are the computational bases of $\cH_{A}$ and $\cH_{B}$, respectively. For any bipartite state $\ket{\psi}\in \cH_{A}\otimes \cH_{B}$, it can be expressed by
		\begin{equation}
		    \ket{\psi}=\sum_{i\in \bbZ_m, j\in\bbZ_n}a_{i,j}\ket{i}_A\ket{j}_B.
		\end{equation}
		Then $\ket{\psi}$ corresponds to an $m\times n$ matrix 
		\begin{equation}
		    M=(a_{i,j})_{i\in\bbZ_m,j\in\bbZ_n}.
		\end{equation}
		If $\rank(M)=1$, then  $\ket{\psi}$ is a product state; if $\rank(M)\geq 2$, then  $\ket{\psi}$ is an entangled state. Assume $\ket{\psi_i}\in \cH_{A}\otimes \cH_{B}$ corresponds to an $m\times n$ matrix $M_i$ for $i=1,2$, then $\braket{\psi_1}{\psi_2}=\tr({M_1^{\dagger}}{M_2})$.
		Let $\cH=\otimes_{i=1}^n \cH_i$ be an $n$-partite Hilbert space. 
		An orthogonal product set (OPS) in $\cH$ is a set of orthogonal product states, and an orthogonal product basis (OPB) in $\cH$ is an OPS which spans $\cH$. 
		Given $\cH=\otimes_{i=1}^n \cH_i$, let $\cH_{ext}=\otimes_{i=1}^n(\cH_i\oplus\cH_i')$ be a locally extended Hilbert space of $\cH$, where $H_i'$ is a local extension.
		Now, we review some definitions. 		
		
     \begin{definition}\label{def:ucpb}
      Let $\cS$ be an OPS in $\cH=\otimes_{i=1}^n \cH_i$.  The set $\cS$ spans a subspace $\cH_{\cS}$ in $\cH$, and $\dim(\cH_{\cS})<\dim(\cH)$.
    If the complementary subspace $\cH_{\cS}^{\bot}$ contains no product state, then $\cS$ is called an \emph{unextendible product basis (UPB)}.
    If $\cS$ cannot be extended to an OPB in $\cH$, then $\cS$ is called an \emph{uncompletable product basis (UCPB)}.
        Moreover, if $\cS$ is a UCPB in any locally extended Hilbert space $\cH_{ext}=\otimes_{i=1}^n(\cH_i\oplus\cH_i')$, then $\cS$ is called a \emph{strongly uncompletable product basis (SUCPB)}.
     \end{definition}

From Definition~\ref{def:ucpb}, a UPB or an SUCPB must be a UCPB. We can always obtain a UPB from a UCPB $\cS$, by adding some orthogonal product states to $\cS$ from $\cH_{\cS}^{\bot}$ till the new OPS is a UPB.  Moreover, UPBs and SUCPBs cannot be perfectly distinguished under local POVMs and classical communication, and UCPBs can not be perfectly distinguished under local projective measurements and classical communication \cite{bennett1999unextendible}.

For an OPS $\cS=\{\ket{\psi_i}\}_{i=1}^s$ in $\cH=\otimes_{i=1}^n \cH_i$ with $\dim(\cH)=D$ (where $s<D$), we can define a mixed state that is proportional to the projector on  $\cH_{\cS}^{\bot}$,
\begin{equation}
    \overline{\rho}_{\cS}=\frac{1}{D-s}\left(\bbI-\sum_{i=1}^s\ketbra{\psi_i}{\psi_i}\right).
\end{equation}
 Applying partial transposition map to $\overline{\rho}_{\cS}$ in any bipartition, then we can find that $(\bbI\otimes T)\overline{\rho}_{\cS}\geq 0$. It means that $\overline{\rho}_{\cS}$ has the positive partial transpose (PPT) property in any bipartition. If $\cS$ is a UPB, then $\overline{\rho}_{\cS}$ must be entangled from the definition. Thus $\overline{\rho}_{\cS}$ is a PPT entangled state, which is also a bound entangled state (no pure entanglement can be distilled) \cite{bennett1999unextendible,divincenzo2003unextendible}.
However, if $\cS$ is a UCPB or an SUCPB, $\overline{\rho}_{\cS}$ is either separable or entangled \cite{bennett1999unextendible,divincenzo2003unextendible}.

It is difficult to show that an OPS is an SUCPB from the definition. This exists a sufficient condition.

\begin{lemma}\label{lem:sucpb}
Let $\cS$ be an OPS in $\cH=\otimes_{i=1}^n \cH_i$. If all the product states in $\cH_{\cS}^{\bot}$ cannot span $\cH_{\cS}^{\bot}$, then $\cS$ is an SUCPB.
\end{lemma}
\begin{proof}
If all the product states in $\cH_{\cS}^{\bot}$ cannot span $\cH_{\cS}^{\bot}$, then $\overline{\rho}_{\cS}$ must be entangled by Theorem 2(ii) in \cite{horodecki1997separability}. Further, according to Proposition 1 in \cite{divincenzo2003unextendible}, $\cS$ is an SUCPB.
\end{proof}
\vspace{0.4cm}

	\begin{figure}[t]
		\centering
		\includegraphics[scale=0.5]{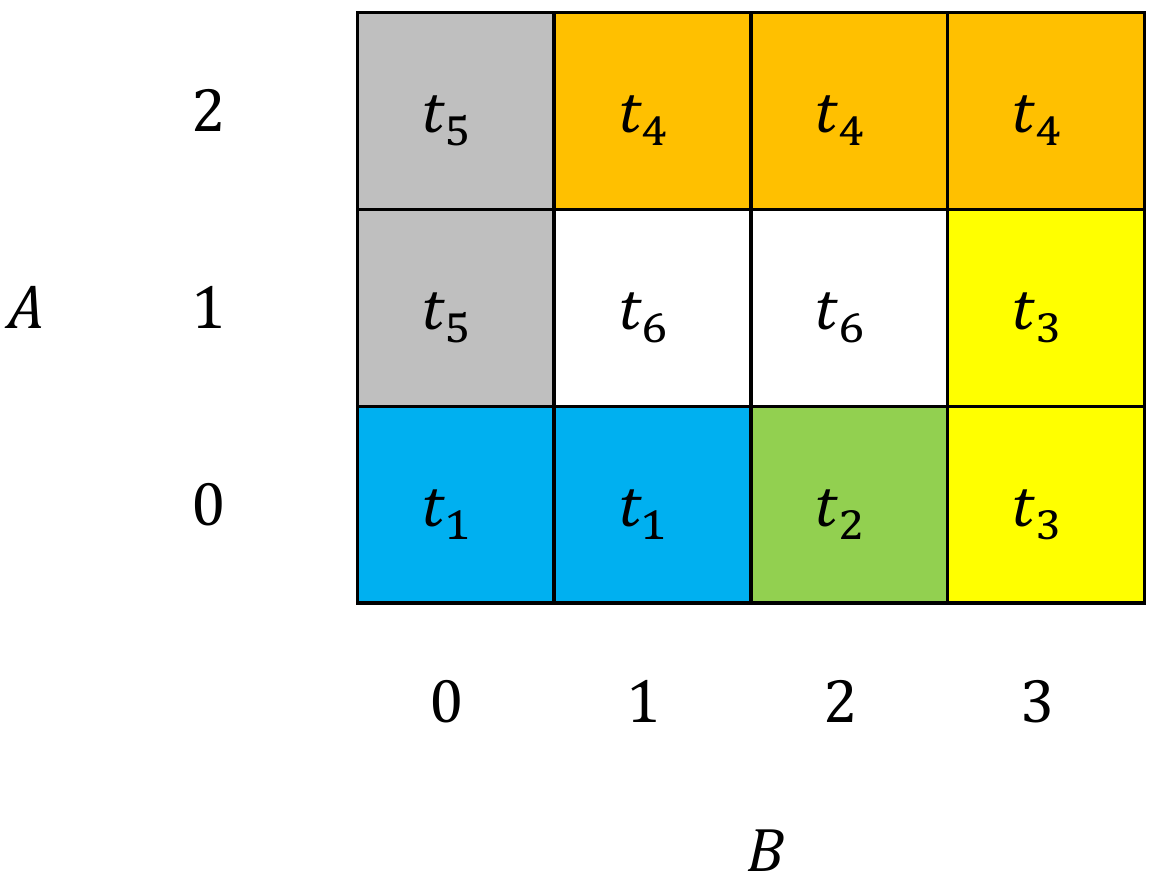}
		\caption{Tile structure with $6$ tiles in $\bbC^3\otimes \bbC^4$.   } \label{fig:34}
	\end{figure}

By Lemma~\ref{lem:sucpb}, a UPB must be an SUCPB. However, the converse is not true. We will give an example of SUCPB, which is not a UPB.

Tile structures can be used to construct UPBs \cite{divincenzo2003unextendible,halder2019family,shi2020unextendible}. Next, we show that tile structures can also be used to construct SUCPBs.  A tile structure $\cT$ in $\bbC^m\otimes \bbC^n$ is an $m\times n$ rectangle, which can be partitioned into $s$ disjoint tiles $\{t_i\}_{i=1}^s$. Each tile $t_i$ is a rectangle. We denote $\cT:=\cup_{i=1}^s t_i$. For example, Fig.~\ref{fig:34} gives a tile structure $\cT=\cup_{i=1}^6 t_i$ with  $6$ tiles in $\bbC^3\otimes \bbC^4$. Any tile $t_i$ of $\cT=\cup_{i=1}^s t_i$  has row coordinates $\{p_0,p_1,\ldots,p_{k-1}\}_A$ and column coordinates $\{q_0,q_1,\ldots,q_{\ell-1}\}_B$, and we denote it as  $t_i=\{p_0,p_1,\ldots,p_{k-1}\}_A\times \{q_0,q_1,\ldots,q_{\ell-1}\}_B$, where $\{p_0,p_1,\ldots,p_{k-1}\}$ and $\{q_0,q_1,\ldots,q_{\ell-1}\}$ are subsets of $\bbZ_m$ and $\bbZ_n$, respectively. For tile $t_i$, we can construct an OPS of size $k\ell$ in $\bbC^m\otimes \bbC^n$,
\begin{equation}\label{eq:A_i}
\begin{aligned}
    \cA_i=\{&\ket{\psi_i{(a,b)}}:=\left(\sum_{e\in\bbZ_k}m_{a,e}\ket{p_e}\right)_A\left(\sum_{e\in\bbZ_\ell}n_{b,e}\ket{q_e}\right)_B\\
    &\mid (a,b)\in\bbZ_k\times\bbZ_\ell\}.
\end{aligned}
\end{equation}
Here the coefficient matrix $M=(m_{a,e})_{a,e\in\bbZ_k}$ is a  $k\times k$ row orthogonal matrix (row vectors are mutually orthogonal), and $m_{0,e}=1$ for $e\in \bbZ_k$, and the coefficient matrix $N=(n_{b,e})_{b,e\in\bbZ_\ell}$ is an  $\ell\times \ell$ row orthogonal matrix, and $n_{0,e}=1$ for $e\in \bbZ_\ell$. For example, we can choose $M=(w_k^{ae})_{a,e\in\bbZ_k}$, and $N=(w_\ell^{be})_{b,e\in\bbZ_\ell}$.
 Since those tiles in $\cT=\cup_{i=1}^s t_i$ are disjoint, we can obtain an OPB 
 \begin{equation}
    \cB:=\cup_{i=1}^s\cA_i  
 \end{equation}
 in $\bbC^m\otimes \bbC^n$. Further, we define the ``stopper" state as
\begin{equation}\label{eq:mn_S}
    \ket{S}=\left(\sum_{i\in\bbZ_m}\ket{i}\right)_A\left(\sum_{j\in\bbZ_n}\ket{j}\right)_B.
\end{equation}
We mainly consider the following OPS,
\begin{equation}\label{eq:St}
   \cS:=\cup_{i=1}^s(\cA_{i}\setminus\{\ket{\psi_i{(0,0)}}\})\cup\{\ket{S}\}. 
\end{equation}

For example, using the tile structure $\cT=\cup_{i=1}^6 t_i$ in  Fig.~\ref{fig:34}, we  obtain an OPB $\cB=\cup_{i=1}^6\cA_i$ in $\bbC^3\otimes \bbC^4$, where
\begin{equation}\label{eq:34}
\begin{aligned}
    \cA_1&=\{\ket{\psi_1{(0,b)}}=\ket{0}_A(\ket{0}+(-1)^{b}\ket{1})_B\mid b\in\bbZ_2\},\\
    \cA_2&=\{\ket{\psi_2{(0,0)}}=\ket{0}_A\ket{2}_B\},\\
    \cA_3&=\{\ket{\psi_3{(a,0)}}=(\ket{0}+(-1)^{a}\ket{1})_A\ket{3}_B\mid a\in\bbZ_2\},\\
    \cA_4&=\{\ket{\psi_4{(0,b)}}=\ket{2}_A(\ket{1}+w_3^{b}\ket{2}+w_3^{2b}\ket{3})_B\mid b\in\bbZ_3\},\\
    \cA_5&=\{\ket{\psi_5{(a,0)}}=(\ket{1}+(-1)^{a}\ket{2})_A\ket{0}_B\mid a\in\bbZ_2\},\\
    \cA_6&=\{\ket{\psi_6{(0,b)}}=\ket{1}_A(\ket{1}+(-1)^{b}\ket{2})_B\mid b\in\bbZ_2\}.
    \end{aligned}
\end{equation}
The ``stopper" state is
\begin{equation}\label{eq:34_S}
    \ket{S}=(\ket{0}+\ket{1}+\ket{2})_A(\ket{0}+\ket{1}+\ket{2}+\ket{3})_B.
\end{equation}
Next, we show that
\begin{equation}\label{eq:S34}
   \cS:=\cup_{i=1}^6(\cA_{i}\setminus\{\ket{\psi_i{(0,0)}}\})\cup\{\ket{S}\} 
\end{equation}
is an SUCPB in $\bbC^3\otimes \bbC^4$.
\begin{example}\label{example:34}
In $\bbC^3\otimes \bbC^4$, the OPS $\cS$ given by Eq.~\eqref{eq:S34} is an SUCPB.
\end{example}
\begin{proof}
Let $\cS_1:=\cup_{i=1}^6(\cA_{i}\setminus\{\ket{\psi_i{(0,0)}}\})$ and $\cS_2:=\cup_{i=1}^6\ket{\psi_i{(0,0)}}$. We know that $\cS_1\cup \cS_2$ is an OPB in $\bbC^3\otimes \bbC^4$. Since $\cH_{\cS_1}\subset \cH_{\cS}$, it implies $\cH_{\cS}^{\bot}\subset \cH_{\cS_1}^{\bot}=\cH_{\cS_2}$. Then for any product state $\ket{\psi}\in \cH_{\cS}^{\bot}$, there exists $a_i\in \bbC$ for $1\leq i\leq 6$, such that
\begin{equation*}
    \ket{\psi}=\sum_{i=1}^6a_i\ket{\psi_i{(0,0)}}.
\end{equation*}
Next, $\ket{\psi}$ corresponds to a $3\times 4$ matrix,
\begin{equation*}
    M=\begin{pmatrix}
    a_5 & a_4 & a_4 & a_4\\
    a_5 & a_6 & a_6 & a_3\\
    a_1 & a_1 & a_2 & a_3\\
    \end{pmatrix}.
\end{equation*}
Note that $M$ has a similar structure to the tile structure in Fig.~\ref{fig:34}.
The ``stopper" state $\ket{S}$ corresponds to a all-ones matrix $J$, where every element is equal to one.
Since $\ket{\psi}$ is a product state and $\braket{S}{\psi}=0$, we have $\rank(M)=1$ and $\text{sum}(M)=0$. This is only possible for
\begin{equation*}
    M=\begin{pmatrix}
    0 & 0 & 0 & 0\\
    0 & 0 & 0 & 0\\
    a_1 & a_1 & a_2 &0\\
    \end{pmatrix}, \quad 2a_{1}+a_2=0.
\end{equation*}
It means that $\cH_{\cS}^{\bot}$ contains only one product state $\ket{0}(\ket{0}+\ket{1}-2
\ket{2})$. Since $\dim(\cH_{\cS}^{\bot})=5$, the OPS $\cS$ is an SUCPB by Lemma~\ref{lem:sucpb}.
\end{proof}
\vspace{0.4cm}

Since there exists a  product state $\ket{\psi}=\ket{0}(\ket{0}+\ket{1}-2
\ket{2})\in \cH_{\cS}^{\bot}$, $\cS$ is not a UPB.  However, if we add $\ket{\psi}$  to $\cS$, then $\cS'=\cS\cup\{\ket{\psi}\}$ must be a UPB in $\bbC^3\otimes \bbC^4$. In fact, for any  product state $\ket{\phi}\in \cH_{\cS'}^{\bot}$,  $\ket{\phi}$ corresponds to a $3\times 4$ matrix,
\begin{equation*}
    M=\begin{pmatrix}
    a_5 & a_4 & a_4 & a_4\\
    a_5 & a_6 & a_6 & a_3\\
    a_1 & a_1 & a_2 & a_3\\
    \end{pmatrix},
\end{equation*}
where $2a_1=a_2$, $\rank(M)=1$ and $\text{sum}(M)=0$. Such a matrix $M$ does not exist.
From the above discussion,  we can obtain a sufficient condition for the construction of UPBs by tile structures.
\begin{lemma}\label{lem:upb}
For a tile structure $\cT=\cup_{i=1}^st_i$ ($s\geq 5$) in $\bbC^m\otimes \bbC^n$, if any $r$ ($2\leq r\leq s-1$) tiles cannot form a rectangle, then the OPS $\cS$ given by Eq.~\eqref{eq:St} is a UPB in $\bbC^m\otimes \bbC^n$.
\end{lemma}


The tile structure in Lemma~\ref{lem:upb} is the U-tile structure proposed in Ref.~\cite{shi2020unextendible}. In Ref.~\cite{bennett1999unextendible}, the authors gave a UCPB, which is not an SUCPB. Let $\cD(\text{UPB})$ be the set of all UPBs; $\cD(\text{SUCPB})$ be the set of all SUCPBs; and $\cD(\text{UCPB})$ be the set of all UCPBs. Then following set inclusion relation is obtained,
\begin{equation*}
    \cD(\text{UPB})\subsetneq  \cD(\text{SUCPB})\subsetneq \cD(\text{UCPB}),
\end{equation*}
See also Fig.~\ref{fig:inclusion} for the inclusion relation of these three sets.

	\begin{figure}[t]
		\centering
		\includegraphics[scale=0.35]{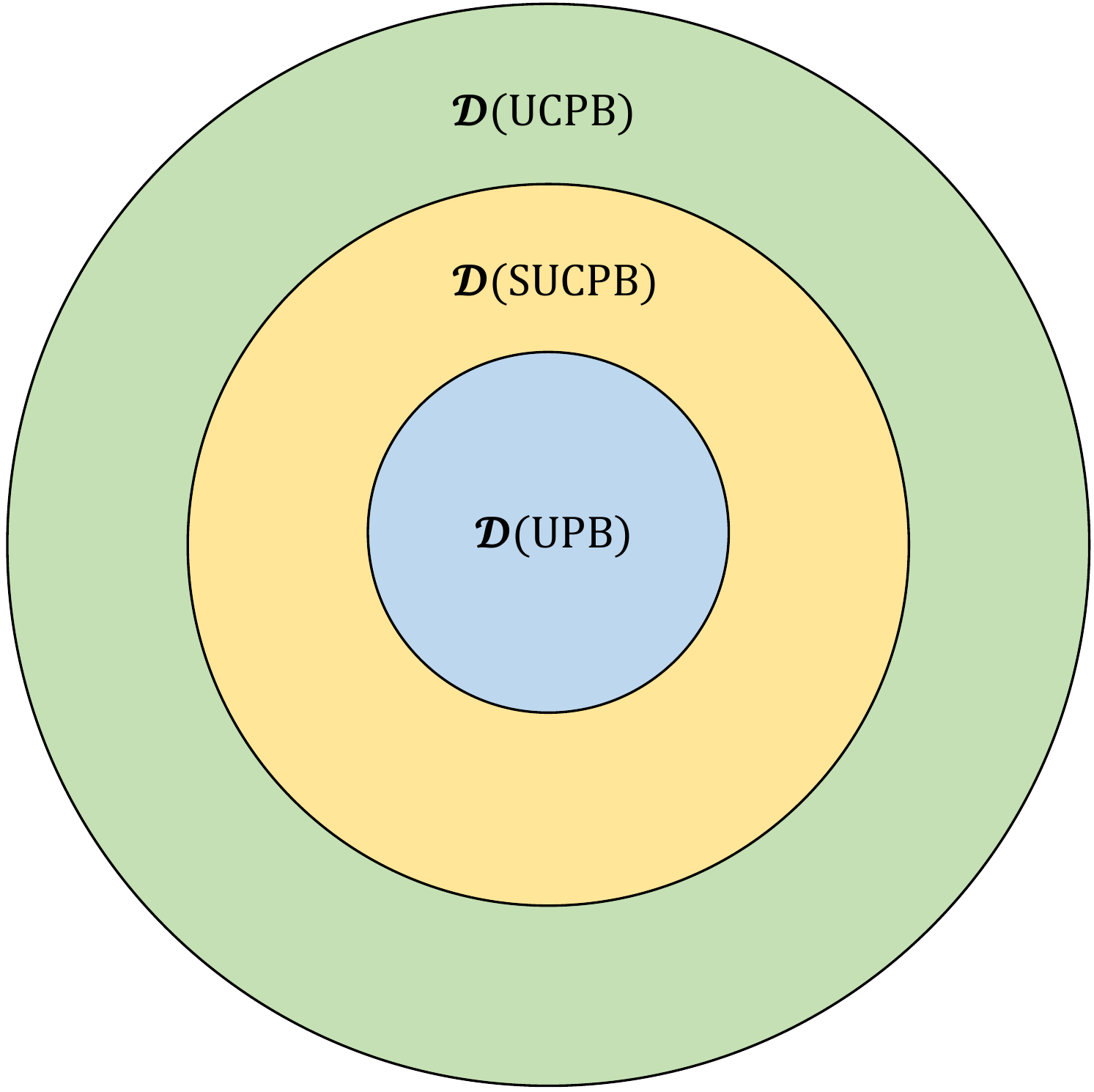}
		\caption{A set inclusion relation among the set of all UPBs  $\cD(\text{UPB})$,  the set of all SUCPBs  $\cD(\text{SUCPB})$, and the set of all UCPBs $\cD(\text{UCPB})$. The set $\cD(\text{UPB})$ is a proper subset of $\cD(\text{SUCPB})$, and $\cD(\text{SUCPB})$ is a proper subset of $\cD(\text{UCPB})$.   } \label{fig:inclusion}
	\end{figure}

In Ref.~\cite{divincenzo2003unextendible}, the authors proposed an open question: can we find a UPB which is a UCPB in every bipartition? We will give a positive answer, by showing a stronger UPB, which is an SUCPB in every bipartition.

\section{The existence of a UPB that is an SUCPB in every bipartition}\label{sec:sucpb}
In this section, we show that there exists a UPB which is an SUCPB in every bipartition in any three, and four-partite system. Since any OPS in $\bbC^2 \otimes \bbC^n$ can be extended to an OPB \cite{bennett1999unextendible,divincenzo2003unextendible}, the minimum system for the existence of such UPBs is $\bbC^3\otimes \bbC^3\otimes \bbC^3$. The following UPB in $\bbC^3\otimes \bbC^3\otimes \bbC^3$ is from \cite{agrawal2019genuinely}, which is constructed from the tile structure in tripartite system (we will introduce tile structures in multipartite systems in Sec.~\ref{sec:gupb}).
Consider an OPB $\cup_{i=1}^9\cA_i$ in  $\bbC^3\otimes \bbC^3\otimes \bbC^3$,
\begin{equation}\label{eq:opb333}
\begin{aligned}
    \cA_1:=&\{\ket{\psi_1(i,j)}=\ket{\xi_i}_A\ket{0}_B\ket{\eta_k}_C\mid (i,j)\in\bbZ_2\times \bbZ_2\},\\
    \cA_2:=&\{\ket{\psi_2(i,j)}=\ket{\xi_i}_A\ket{\eta_j}_B\ket{2}_C\mid (i,j)\in\bbZ_2\times \bbZ_2\},\\ \cA_3:=&\{\ket{\psi_3(i,j)}=\ket{2}_A\ket{\xi_i}_B\ket{\eta_j}_C\mid (i,j)\in\bbZ_2\times \bbZ_2\},\\
    \cA_4:=&\{\ket{\psi_4(i,j)}=\ket{\eta_i}_A\ket{2}_B\ket{\xi_j}_C\mid (i,j)\in\bbZ_2\times \bbZ_2\},\\
    \cA_5:=&\{\ket{\psi_5(i,j)}=\ket{\eta_i}_A\ket{\xi_j}_B\ket{0}_C\mid (i,j)\in\bbZ_2\times \bbZ_2\},\\ \cA_6:=&\{\ket{\psi_6(i,j)}=\ket{0}_A\ket{\eta_i}_B\ket{\xi_j}_C\mid (i,j)\in\bbZ_2\times \bbZ_2\},\\
    \cA_7:=&\{\ket{0}_A\ket{0}_B\ket{0}_C\},\\
    \cA_8:=&\{\ket{1}_A\ket{1}_B\ket{1}_C\},\\
    \cA_9:=&\{\ket{2}_A\ket{2}_B\ket{2}_C\},
\end{aligned}
\end{equation}
where $\ket{\eta_s}_X=\ket{0}_X+(-1)^s\ket{1}_X$, $\ket{\xi_s}_X=\ket{1}_X+(-1)^s\ket{2}_X$ for $s\in\bbZ_2$, and $X\in\{A,B,C\}$. The  ``stopper" state is,
\begin{equation}\label{eq:333_S}
    \ket{S}=(\ket{0}+\ket{1}+\ket{2})_A(\ket{0}+\ket{1}+\ket{2})_B(\ket{0}+\ket{1}+\ket{2})_C.
\end{equation}
Then 
\begin{equation}\label{eq:U333}
  \cU:=\cup_{i=1}^6(\cA_i \setminus\{\ket{\psi_i(0,0)}\})\cup\ket{S}  
\end{equation}
is a UPB in $\bbC^3\otimes \bbC^3\otimes \bbC^3$  \cite{agrawal2019genuinely}. Now, we have the following lemma.

	\begin{figure}[t]
		\centering
		\includegraphics[scale=0.4]{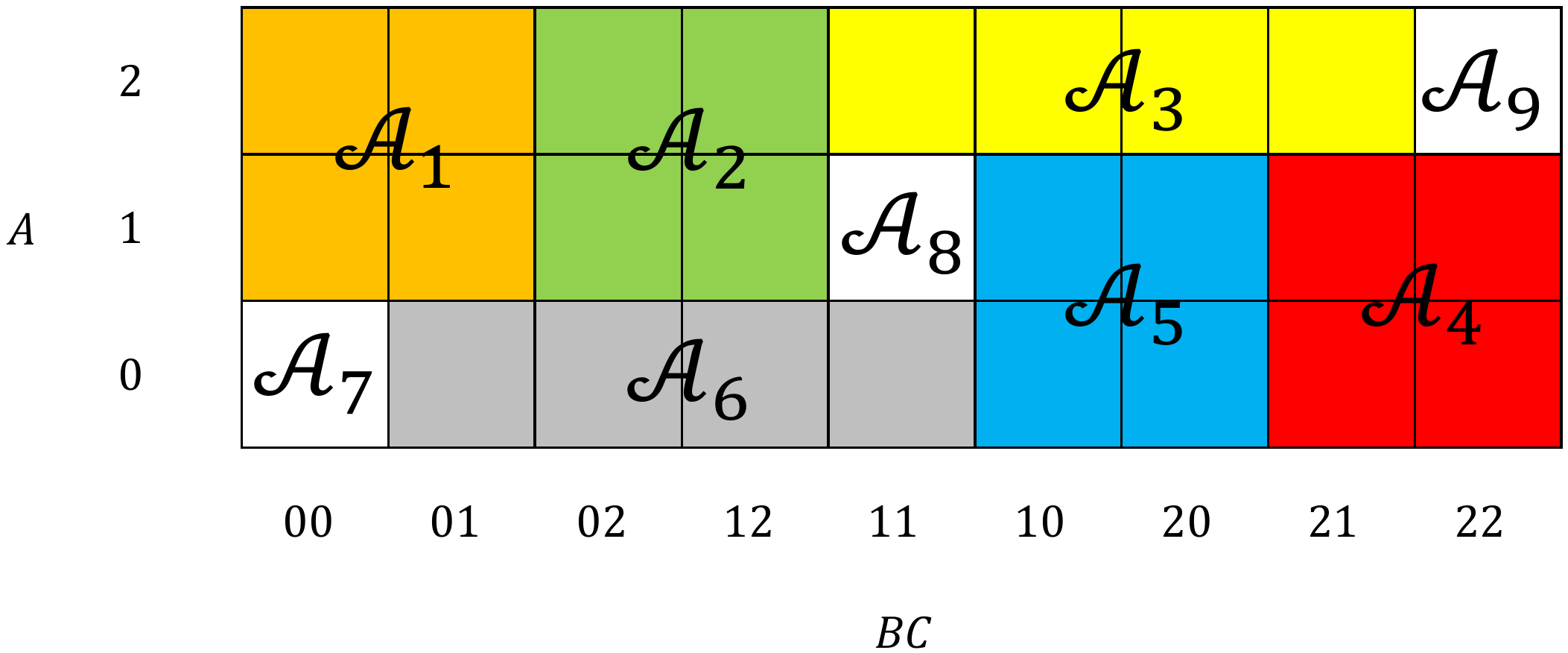}
		\caption{The corresponding tile structure in $\bbC^3\otimes \bbC^9$ of the OPB $\cup_{i=1}^9\cA_i$ (Eq.~\eqref{eq:opb333})   in $A|BC$ bipartition.  } \label{fig:tile39}
	\end{figure}

\begin{lemma}\label{lem:333}
In $\bbC^3\otimes \bbC^3\otimes \bbC^3$, the UPB  $\cU$ given by Eq.~\eqref{eq:U333} is an SUCPB in every bipartition.
\end{lemma}
\begin{proof}
First, we consider the bipartition $A|BC$. The OPB $\cup_{i=1}^9\cA_i$ given by Eq.~\eqref{eq:opb333}  in $A|BC$ bipartition corresponds the tile structure in Fig.~\ref{fig:tile39}. Next, we show that the OPS $\cU_{A|BC}$ is an SUCPB in $\bbC^3\otimes \bbC^9$. For the same discussion as Example~\ref{example:34}, we can assume that $\ket{\psi}\in\cH_{\cU_{A|BC}}^{\bot}$ is a product state. By Fig.~\ref{fig:tile39}, $\ket{\psi}$ corresponds to a $3\times 9$ matrix
\begin{equation*}
   M= \begin{pmatrix}
  a_1 &a_1 &a_2 &a_2 &a_3 &a_3 &a_3 &a_3 &a_9\\
  a_1 &a_1 &a_2 &a_2 &a_8 &a_5 &a_5 &a_4 &a_4\\
  a_7 &a_6 &a_6 &a_6 &a_6 &a_5 &a_5 &a_4 &a_4\\
    \end{pmatrix},
\end{equation*}
where $a_i\in\bbC$ for $1\leq i\leq 9$, $\rank(M)=1$, and $\text{sum}(M)=0$. There are only four cases,
\begin{enumerate}[(i)]
    \item $a_1+a_2=0$,  $a_i=0$ for $1\leq i\leq 9$ and $i\neq 1,2$;
    \item $4a_3+a_9=0$, $a_i=0$ for $1\leq i\leq 9$ and $i\neq 3,9$;
    \item $a_4+a_5=0$,  $a_i=0$ for $1\leq i\leq 9$ and $i\neq 4,5$;
    \item $4a_6+a_7=0$, $a_i=0$ for $1\leq i\leq 9$ and $i\neq 6,7$.
\end{enumerate}
It means that there are only four product states in $\cH_{\cU_{A|BC}}^{\bot}$: $(\ket{1}+\ket{2})_A(\ket{00}+\ket{01}-\ket{02}-\ket{12})_{BC}$,
$\ket{2}_A(\ket{11}+\ket{10}+\ket{20}+\ket{21}-4\ket{22})_{BC}$,
$(\ket{0}+\ket{1})_A(\ket{10}+\ket{20}-\ket{21}-\ket{22})_{BC}$,
and $\ket{0}_A(\ket{01}+\ket{02}+\ket{12}+\ket{11}-4\ket{00})_{BC}$.
Since $\dim(\cH_{\cU_{A|BC}}^{\bot})=8$, the OPS $\cU_{A|BC}$ is an SUCPB by Lemma~\ref{lem:sucpb}.

Further, since the OPB $\cup_{i=1}^9\cA_i$ given by Eq.~\eqref{eq:opb333}  in any bipartition of $\{A|BC, B|AC, C|AB\}$ corresponds to a similar tile structure in Fig.~\ref{fig:tile39}, we obtain that $\cU_{A|BC}$,  $\cU_{B|AC}$, and $\cU_{C|AB}$ are all SUCPBs. Thus $\cU$ is an SUCPB in every bipartition.
\end{proof}
\vspace{0.4cm}

A similar construction of UPB in $\bbC^{d_1}\otimes \bbC^{d_2}\otimes \bbC^{d_3}$ for $d_1,d_2,d_3\geq 3$ was given in \cite{shi2021strong}. For the same discussion as Lemma~\ref{lem:333}, we have the following theorem.

\begin{theorem}\label{thm:333}
In $\bbC^{d_1}\otimes \bbC^{d_2}\otimes \bbC^{d_3}$, $d_1,d_2,d_3\geq 3$, there exists a UPB which is an SUCPB in every bipartition.
\end{theorem}

Next, we consider the four-partite UPB. The following UPB in $\bbC^3\otimes \bbC^3\otimes \bbC^3\otimes \bbC^3$ was given in \cite{shi2021strong}, which is constructed from the tile structure in four-partite system. Consider an OPB in $\bbC^3\otimes \bbC^3\otimes \bbC^3\otimes \bbC^3$,

   	\begin{equation*}
		\begin{aligned}
			\cA_1:=\{&\ket{\psi_1(i,j,k)}=\ket{\xi_i}_A\ket{\eta_j}_B|0\rangle_C|\xi_k\rangle_D\\
			&\mid (i,j,k)\in\bbZ_{2}\times \bbZ_{2} \times\bbZ_{2}  \},\\	\cA_2:=\{&\ket{\psi_2(i,j,k)}=\ket{\xi_i}_A\ket{2}_B|\eta_j\rangle_C|\eta_k\rangle_D\\
			&\mid      (i,j,k)\in\bbZ_{2}\times \bbZ_{2} \times\bbZ_{2}  \},\\						
			\cA_3:=\{&\ket{\psi_3(i,j,k)}=\ket{\xi_i}_A\ket{\xi_j}_B|\xi_k\rangle_C|2\rangle_D\\
			&\mid (i,j,k)\in\bbZ_{2}\times \bbZ_{2} \times\bbZ_{2}  \},\\
			\cA_4:=\{&\ket{\psi_4(i,j,k)}=\ket{\xi_i}_A\ket{2}_B\ket{0}_C\ket{2}_D\mid i\in\bbZ_{2}\},\\
			\cA_5:=\{&\ket{\psi_5(i,j,k)}=\ket{2}_A\ket{\eta_i}_B|\xi_j\rangle_C|\eta_k\rangle_D\\
			&\mid (i,j,k)\in\bbZ_{2}\times \bbZ_{2} \times\bbZ_{2}  \},\\	
			\cA_6:=\{&\ket{\psi_6(i,j,k)}=\ket{2}_A\ket{\eta_i}_B\ket{0}_C\ket{0}_D\mid i\in\bbZ_{2}\},\\
			\cA_7:=\{&\ket{\psi_7(i,j,k)}=\ket{2}_A\ket{0}_B\ket{\xi_i}_C\ket{2}_D\mid i\in\bbZ_{2}\},\\
			\cA_8:=\{&\ket{\psi_8(i,j,k)}=\ket{2}_A\ket{2}_B\ket{2}_C\ket{\eta_i}_D\mid i\in\bbZ_{2}\},\\
			\cA_9:=\{&\ket{\psi_9(i,j,k)}=\ket{\eta_i}_A\ket{\xi_j}_B|2\rangle_C|\eta_k\rangle_D\\
			&\mid (i,j,k)\in\bbZ_{2}\times \bbZ_{2} \times\bbZ_{2} \},\\
			\cA_{10}:=\{&\ket{\psi_{10}(i,j,k)}=\ket{\eta_i}_A\ket{0}_B|\xi_j\rangle_C|\xi_k\rangle_D\\&\mid (i,j,k)\in\bbZ_{2}\times \bbZ_{2} \times\bbZ_{2}  \},\\	
      \end{aligned}
      \end{equation*}
         	\begin{equation}\label{eq:opb3333}
		\begin{aligned}
			\cA_{11}:=\{&\ket{\psi_{11}(i,j,k)}=\ket{\eta_i}_A\ket{\eta_j}_B|\eta_k\rangle_C|0\rangle_D\\&\mid (i,j,k)\in\bbZ_{2}\times \bbZ_{2} \times\bbZ_{2}   \},\\
			\cA_{12}:=\{&\ket{\psi_{12}(i,j,k)}=\ket{\eta_i}_A\ket{0}_B\ket{2}_C\ket{0}_D\mid i\in\bbZ_{2}\},\\
					\cA_{13}:=\{&\ket{\psi_{13}(i,j,k)}=\ket{0}_A\ket{\xi_i}_B|\eta_j\rangle_C|\xi_k\rangle_D\\
			&\mid (i,j,k)\in\bbZ_{2}\times \bbZ_{2} \times\bbZ_{2}  \},\\
			\cA_{14}:=\{&\ket{\psi_{14}(i,j,k)}=\ket{0}_A\ket{\xi_i}_B\ket{2}_C\ket{2}_D\mid i\in\bbZ_{2}\},\\
			\cA_{15}:=\{&\ket{\psi_{15}(i,j,k)}=\ket{0}_A\ket{2}_B\ket{\eta_i}_C\ket{0}_D\mid i\in\bbZ_{2}\},\\
			\cA_{16}:=\{&\ket{\psi_{16}(i,j,k)}=\ket{0}_A\ket{0}_B\ket{0}_C\ket{\xi_i}_D\mid i\in\bbZ_{2}\},\\
			\cA_{17}:=\{&\ket{1}_A\ket{1}_B\ket{1}_C\ket{1}_D\},\\
		\end{aligned}
\end{equation}
where $\ket{\eta_s}_X=\ket{0}_X+(-1)^s\ket{1}_X$, $\ket{\xi_s}_X=\ket{1}_X+(-1)^s\ket{2}_X$ for $s\in\bbZ_2$, and $X\in\{A,B,C,D\}$. The  ``stopper" state is,
\begin{equation}\label{eq:3333_S}
\begin{aligned}
    \ket{S}=&(\ket{0}+\ket{1}+\ket{2})_A(\ket{0}+\ket{1}+\ket{2})_B(\ket{0}+\ket{1}+\ket{2})_C\\
    &(\ket{0}+\ket{1}+\ket{2})_D.
\end{aligned}
\end{equation}
Then
\begin{equation}\label{eq:U3333}
    \cV:=\cup_{i=1}^{16}(\cA_i \setminus\{\ket{\psi_i(0,0,0)}\})\cup\{\ket{S}\} 
\end{equation}
is a UPB in $\bbC^3\otimes \bbC^3\otimes \bbC^3\otimes \bbC^3$ \cite{agrawal2019genuinely}. We can show that this UPB is an SUCPB in every bipartition.
\begin{lemma}\label{lem:3333}
In $\bbC^3\otimes \bbC^3\otimes \bbC^3\otimes \bbC^3$, the UPB  $\cV$ given by Eq.~\eqref{eq:U3333}  is an SUCPB in every bipartition.
\end{lemma}
\begin{proof}
We need to consider the bipartition set $\{A|BCD, B|ACD, C|ABD, D|ABC, AB|CD, AC|BD,$
$ AD|BC\}$. Since the OPB $\cup_{i=1}^{17}\cA_{i}$ given by Eqs.~\eqref{eq:opb3333} in any bipartition of $\{A|BCD, B|ACD, C|ABD, D|ABC\}$ (or $\{AB|CD, AC|BD,$  $AD|BC\}$) has a similar structure, we only need to consider $\cV_{A|BCD}$ and $\cV_{AB|CD}$.

For $\cV_{A|BCD}$, we assume that $\ket{\psi}\in\cH_{\cV_{A|BCD}}^{\bot}$ is a product state, then  $\ket{\psi}$ corresponds to a $3\times 27$ matrix,
\begin{widetext}
	\begin{figure}[h]
		\centering
		\includegraphics[scale=0.55]{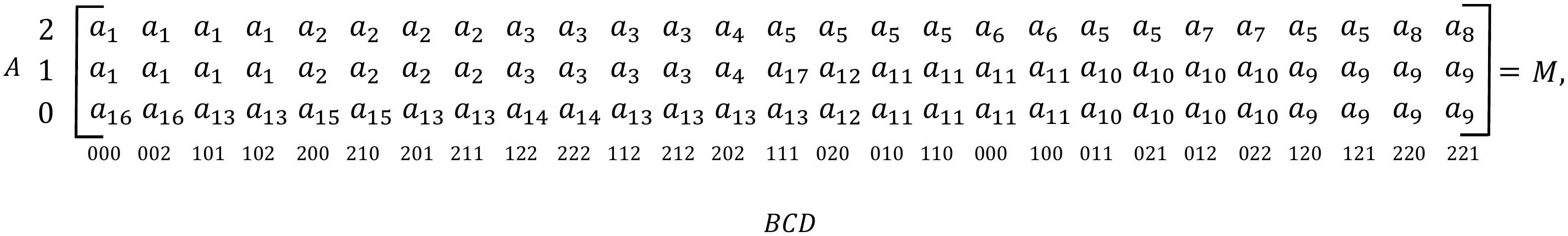}
	\end{figure}
\end{widetext}	
	
 \noindent where $a_i\in\bbC$ for $1\leq i\leq 17$, $\rank(M)=1$, and $\text{sum}(M)=0$. There are only four cases,
 \begin{enumerate}[(i)]
    \item $4a_1+4a_2+4a_3+a_4=0$, and  $a_i=0$ for $1\leq i\leq 17$ and $i\neq 1,2,3,4$;
    \item $4a_5+a_6+a_7+a_8=0$,  and  $a_i=0$ for $1\leq i\leq 17$ and $i\neq 5,6,7,8$;
    \item $4a_9+4a_{10}+4a_{11}+a_{12}=0$, and $a_i=0$ for $1\leq i\leq 17$ and $i\neq 9,10,11,12$;
   \item $4a_{13}+a_{14}+a_{15}+a_{16}=0$, and  $a_i=0$ for $1\leq i\leq 17$ and $i\neq 13,14,15,16$.
\end{enumerate}
Then $\ket{\psi}$ must belong to one of the four subspaces,
\begin{enumerate}[(i)]
    \item $O_1=\{(\ket{1}+\ket{2})_A(a_1(\ket{000}+\ket{002}+\ket{101}+\ket{102})+a_2(\ket{200}+\ket{210}+\ket{201}+\ket{211})+a_3(\ket{122}+\ket{222}+\ket{112}+\ket{212})+a_4\ket{202})_{BCD}\mid 4a_1+4a_2+4a_3+a_4=0\}$;
    \item $O_2=\{(\ket{2})_A(a_5(\ket{111}+\ket{020}+\ket{010}+\ket{110}+\ket{011}+\ket{021}+\ket{120}+\ket{121})+a_6(\ket{000}+\ket{100})+a_7(\ket{012}+\ket{022})+a_8(\ket{220}+\ket{221}))_{BCD}\mid 4a_5+a_6+a_7+a_8=0 \}$;
    \item $O_3=\{(\ket{0}+\ket{1})_A(a_9(\ket{221}+\ket{220}+\ket{121}+\ket{120})+a_{10}(\ket{022}+\ket{012}+\ket{021}+\ket{011})+a_{11}(\ket{100}+\ket{000}+\ket{110}+\ket{010})+a_{12}\ket{020})_{BCD}\mid 4a_9+4a_{10}+4a_{11}+a_{12}=0\}$;
   \item $O_4=\{(\ket{0})_A(a_{13}(\ket{111}+\ket{202}+\ket{212}+\ket{112}+\ket{211}+\ket{201}+\ket{102}+\ket{101})+a_{14}(\ket{222}+\ket{122})+a_{15}(\ket{210}+\ket{200})+a_{16}(\ket{002}+\ket{000}))_{BCD}\mid 4a_{13}+a_{14}+a_{15}+a_{16}=0 \}$,
\end{enumerate}
where $\dim(O_i)=3$ for $1\leq i\leq 4$, and $O_i\bot O_j$ for $1\leq i\neq j\leq 4$.
Then $\dim (O_1+O_2+O_3+O_4)=12$.  Since $\dim(\cH_{\cV_{A|BCD}}^{\bot})=16$, the OPS $\cV_{A|BCD}$ is an SUCPB by Lemma~\ref{lem:sucpb}.

For $\cV_{AB|CD}$, we assume that $\ket{\phi}\in\cH_{\cV_{AB|CD}}^{\bot}$ is a product state, then  $\ket{\phi}$ corresponds to a $9\times 9$ matrix,

	\begin{figure}[h]
		\centering
		\includegraphics[scale=0.6]{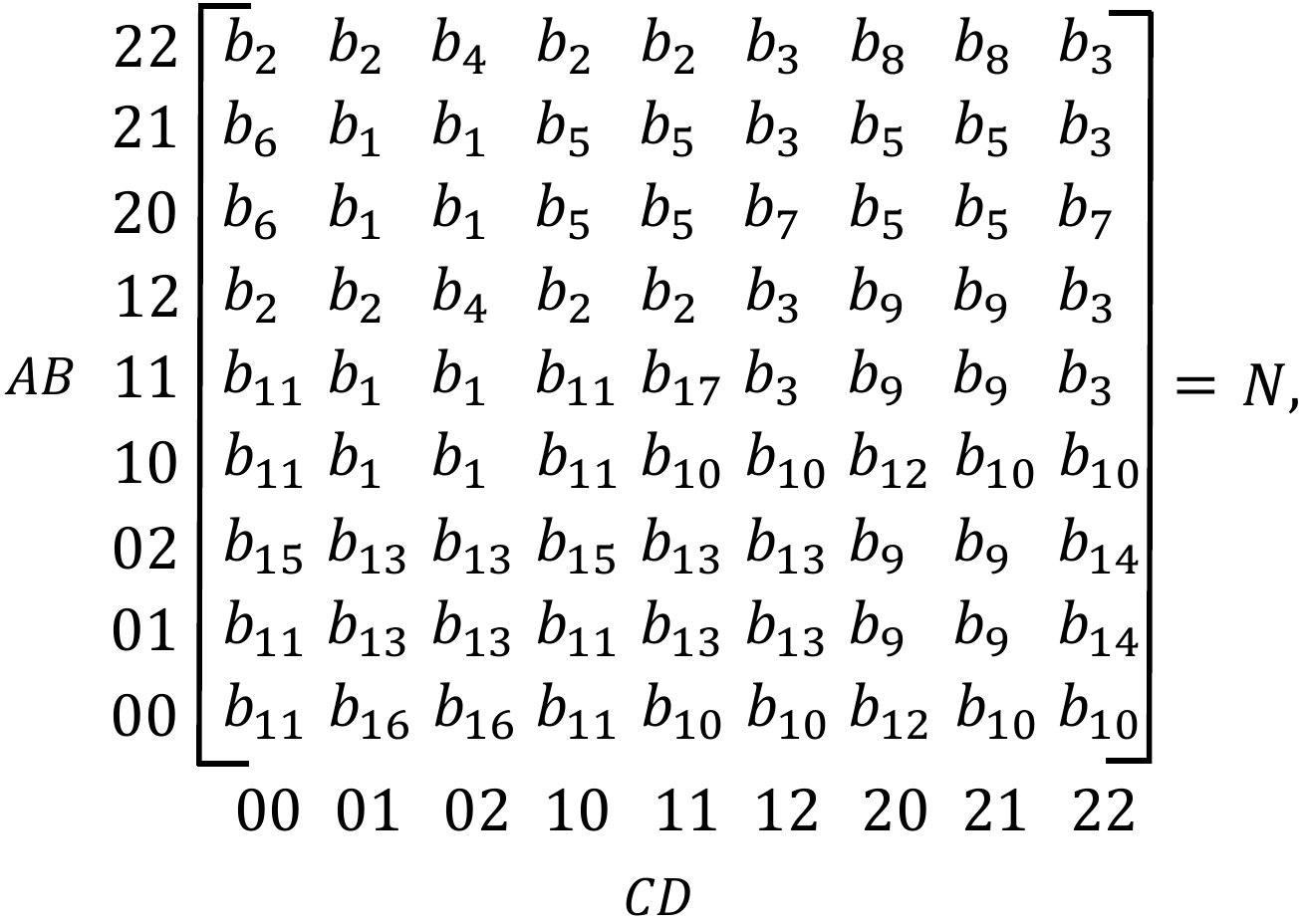}
	\end{figure}

\noindent where $b_i\in\bbC$ for $1\leq i\leq 17$, $\rank(N)=1$, and $\text{sum}(N)=0$. There are only eight cases,
\begin{enumerate}[(i)]
    \item $4b_1+b_{16}=0$, and  $b_i=0$ for $1\leq i\leq 17$ and $i\neq 1,16$;
    \item $4b_2+b_{4}=0$,  and  $b_i=0$ for $1\leq i\leq 17$ and $i\neq 2,4$;
    \item $4b_3+b_{7}=0$,  and  $b_i=0$ for $1\leq i\leq 17$ and $i\neq 3,7$;
    \item $4b_5+b_{6}=0$, and  $b_i=0$ for $1\leq i\leq 17$ and $i\neq 5,6$;
    \item $4b_9+b_{8}=0$, and  $b_i=0$ for $1\leq i\leq 17$ and $i\neq 8,9$;
    \item $4b_{10}+b_{12}=0$,  and  $b_i=0$ for $1\leq i\leq 17$ and $i\neq 10,12$;
    \item $4b_{11}+b_{15}=0$,  and  $b_i=0$ for $1\leq i\leq 17$ and $i\neq 11,15$;
    \item $4b_{13}+b_{14}=0$, and  $b_i=0$ for $1\leq i\leq 17$ and $i\neq 13,14$.
\end{enumerate}
It means that there are only eight product states in $\cH_{\cV_{AB|CD}}^{\bot}$: $(\ket{10}+\ket{11}+\ket{20}+\ket{21}-4\ket{00})_{AB}(\ket{01}+\ket{02})_{CD}$, $(\ket{12}+\ket{22})_{AB}(\ket{00}+\ket{01}+\ket{10}+\ket{11}-4\ket{02})_{CD}$,
$(\ket{11}+\ket{12}+\ket{21}+\ket{22}-4\ket{20})_{AB}(\ket{12}+\ket{22})_{CD}$,
$(\ket{20}+\ket{21})_{AB}(\ket{10}+\ket{11}+\ket{20}+\ket{21}-4\ket{00})_{CD}$,
$(\ket{01}+\ket{02}+\ket{11}+\ket{12}-4\ket{22})_{AB}(\ket{20}+\ket{21})_{CD}$,
$(\ket{00}+\ket{10})_{AB}(\ket{11}+\ket{12}+\ket{21}+\ket{22}-4\ket{20})_{CD}$,
$(\ket{00}+\ket{01}+\ket{10}+\ket{11}-4\ket{02})_{AB}(\ket{00}+\ket{10})_{CD}$,
$(\ket{01}+\ket{02})_{AB}(\ket{01}+\ket{02}+\ket{11}+\ket{12}-4\ket{22})_{CD}$.
Since $\dim(\cH_{\cV_{AB|CD}}^{\bot})=16$, the OPS $\cV_{AB|CD}$ is an SUCPB by Lemma~\ref{lem:sucpb}.

Above all, $\cV$ is an SUCPB in every bipartition.
\end{proof}
\vspace{0.4cm}

The construction of UPB in $\bbC^3\otimes \bbC^3\otimes \bbC^3 \otimes \bbC^3$ was generalized to any four-partite system $\bbC^{d_1}\otimes \bbC^{d_2}\otimes \bbC^{d_3}\otimes \bbC^{d_4}$ for $d_1,d_2,d_3,d_4\geq 3$ \cite{shi2021strong}.
Obviously, for the same discussion as Lemma~\ref{lem:333}, we have the following theorem.

\begin{theorem}\label{thm:3333}
In $\bbC^{d_1}\otimes \bbC^{d_2}\otimes \bbC^{d_3}\otimes \bbC^{d_4}$, $d_1,d_2,d_3,d_4\geq 3$, there exists a UPB which is an SUCPB in every bipartition.
\end{theorem}

However, not all UPBs have this property. For example, let
\begin{equation*}
\begin{aligned}
    \ket{\psi_1}&=\ket{0}_A(\ket{0}-\ket{1})_B,\\
    \ket{\psi_2}&=(\ket{0}-\ket{1})_A\ket{2}_B,\\
    \ket{\psi_3}&=\ket{2}_A(\ket{1}-\ket{2})_B,\\
    \ket{\psi_4}&=(\ket{1}-\ket{2})_A\ket{0}_B,\\
    \ket{\psi_5}&=(\ket{0}+\ket{1}+\ket{2})_A(\ket{0}+\ket{1}+\ket{2})_B,
\end{aligned}
\end{equation*}
then $\cup_{i=1}^5\ket{\psi_i}$ is a UPB in $\bbC^3\otimes \bbC^3$ \cite{divincenzo2003unextendible}. We can construct a UPB in $\bbC^3\otimes \bbC^3\otimes \bbC^3$ from $\cup_{i=1}^5\ket{\psi_i}$  as follows,
\begin{equation*}
\begin{aligned}
\cA_1&:=\{\ket{\psi_i}\ket{0}_C\mid 1\leq i\leq 5\},\\
\cA_2&:=\{\ket{i}_A\ket{j}_B\ket{1}_C\mid i,j\in \bbZ_3\},\\
\cA_3&:=\{\ket{i}_A\ket{j}_B\ket{2}_C\mid i,j\in \bbZ_3\}.\\
\end{aligned}
\end{equation*}
Let $\cW:=\cup_{i=1}^3\cA_i$. For any product state $\ket{\varphi}=\ket{\varphi_1}_A\ket{\varphi_2}_B\ket{\varphi_3}_C\in \cW^{\bot}$, since $\ket{\varphi}$ is orthogonal to any state in $\cA_2\cup \cA_3$,  $\ket{\varphi_3}_C$ must be $\ket{0}_C$. Further, since $\cup_{i=1}^5\ket{\psi_i}$ is a UPB, it means that $\ket{\varphi_1}_A\ket{\varphi_2}_B$ cannot be a product state. Thus $\cW$ is a UPB in $\bbC^3\otimes \bbC^3\otimes \bbC^3$. Nevertheless, $\cW$ is not a UCPB in every bipartition. There must exist four orthogonal states $\{\ket{\psi_6},\ket{\psi_7},\ket{\psi_8},\ket{\psi_9}\}$ such that $\cup_{i=1}^9\ket{\psi_i}$ is an orthogonal basis in $\bbC^3\otimes \bbC^3$. Then $\{\cup_{i=6}^9\ket{\psi_i}\ket{0}_C\}\cup\cW_{AB|C}$ is an OPB in $AB|C$ bipartition. Thus, $\cW$ is not a UCPB in $AB|C$ bipartition.

Next, we consider the application. Note that all UPBs in Theorems~\ref{thm:333} and \ref{thm:3333} cannot be perfectly distinguished under local POVMs and classical communication in any bipartition. These UPBs can be used for local hiding of information \cite{shi2022strong}. For example, assume the information is encoded in the UPB $\cU$ given by Eq.~\eqref{eq:U333}, and the boss send it to his three subordinates: A, B and  C. These three subordinates are from different offices. They can only perform local POVMs, and communicate classic information by telephones.  In this case, the three subordinates cannot obtain the full information, even if any two of them are collusive.   A and B are collusive means that A and B are from the same office and  can perform joint measurements. See also Fig.~\ref{fig:ABC}. Further, Ref.~\cite{shi2021strong} showed a stronger property. Any UPB in Theorems~\ref{thm:333} and \ref{thm:3333} is locally irreducible\footnote{A set of multipartite orthogonal states is
locally irreducible if it is not possible to eliminate one or
more states from the set by orthogonality-preserving local POVMs  \cite{Halder2019Strong}. } in every bipartition, which shows the phenomenon of strong quantum nonlocality without entanglement \cite{Halder2019Strong}.


	\begin{figure}[t]
		\centering
		\includegraphics[scale=0.5]{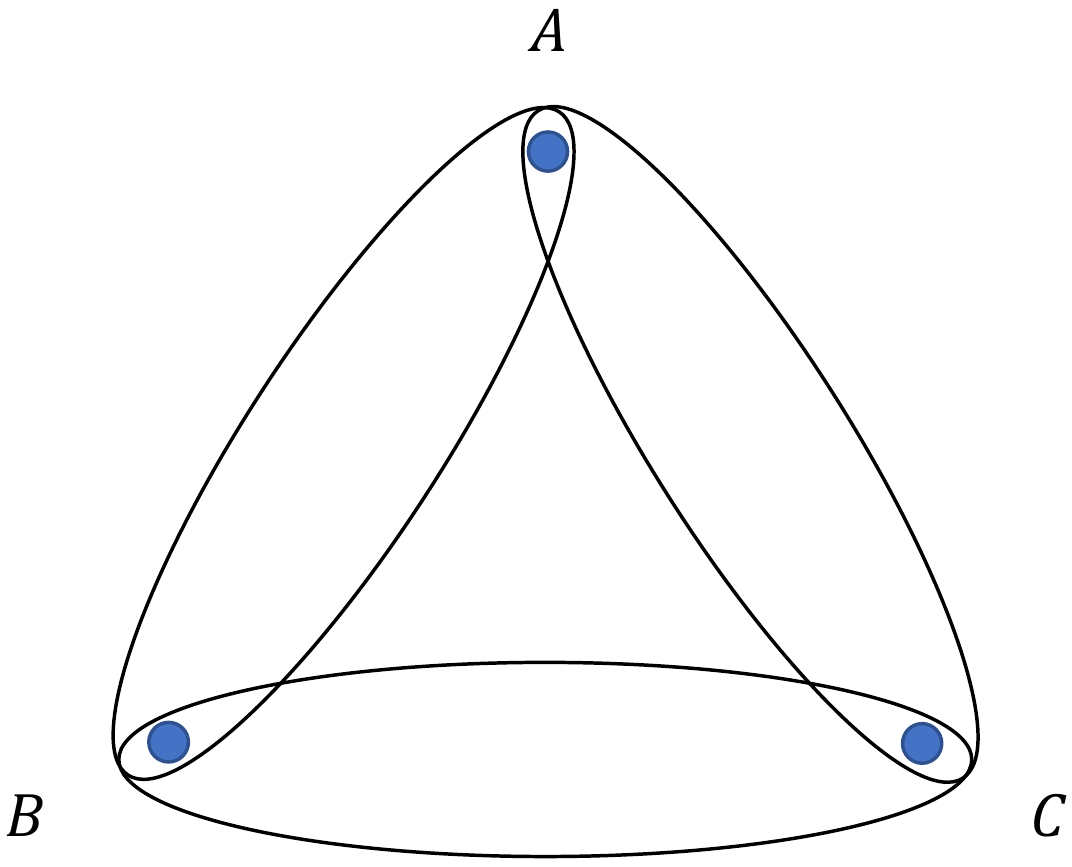}
		\caption{The information is encoded to a tripartite UPB which is an SUCPB is every bipartition, and the boss send it to his three subordinates: A, B, and C. Even if any two of them are collusive, the three subordinates cannot obtain the full information under local POVMs and classical communication.   } \label{fig:ABC}
	\end{figure}

\section{UPBs in every bipartition}\label{sec:gupb}
There exists another open question for UPBs \cite{demianowicz2018unextendible}: can we find a UPB, which is still a UPB in every bipartition? Such a UPB can be used to construct genuinely entangled subspace, and it cannot be perfectly distinguished under local POVMs and classical communication in any bipartition. Unfortunately, any UPB in Sec.~\ref{sec:sucpb} is not a UPB in every bipartition. We will give a sufficient condition for the existence of such a UPB.

We can also generalize the tile structures to multipartite systems.  A tile structure $\cT=\cup_{i=1}^s t_i$ in $\bbC^{d_1}\otimes \bbC^{d_2}\otimes \cdots \otimes \bbC^{d_n}$ is a $d_1\times d_2\times \cdots\times d_n$ hypercube, which can be partitioned into $s$ disjoint tiles $\{t_i\}_{i=1}^s$. Each tile $t_i$ is a hypercube, and it can be expressed by
\begin{equation}
\begin{aligned}
t_i=&\{x_0^{(1)},x_1^{(1)},\ldots,x_{k_1-1}^{(1)}\}_{A_1}\times \{x_0^{(2)},x_1^{(2)},\ldots,x_{k_2-1}^{(2)}\}_{A_2}\\
&\times \cdots\times \{x_0^{(n)},x_1^{(n)},\ldots,x_{k_n-1}^{(n)}\}_{A_n},
\end{aligned}
\end{equation}
where $\{x_0^{(j)},x_1^{(j)},\ldots,x_{k_j-1}^{(j)}\}$ is a subset of $\bbZ_{d_j}$ for $1\leq j\leq n$.
For tile $t_i$, we can construct an OPS of size $k_1k_2\cdots k_n$ in $\bbC^{d_1}\otimes \bbC^{d_2}\otimes \cdots \otimes \bbC^{d_n}$,
\begin{equation}\label{eq:genA_i}
\begin{aligned}
    \cA_i=\{&\otimes_{j=1}^n \left(\sum_{e_j\in\bbZ_{k_j}}m_{a_{j},e_{j}}^{(j)}\ket{x_{e_{j}}^{(j)}}\right)_{A_j}\\
    &\mid a_{j}\in\bbZ_{k_j}, 1\leq j\leq n\},
\end{aligned}
\end{equation}
where each coefficient matrix $M^{(j)}=(m_{a_{j},e_{j}}^{(j)})_{a_{j},e_{j}\in \bbZ_{k_j}}$ is a  $k_j\times k_j$ row orthogonal matrix, and $m_{0,e_{j}}^{(j)}=1$ for $e_j\in\bbZ_{k_j}$. Then we obtain an OPB $\cB:=\cup_{i=1}^s\cA_i$ in $\bbC^{d_1}\otimes \bbC^{d_2}\otimes \cdots \otimes \bbC^{d_n}$.
Note that 
\begin{equation}\label{eq:genpsi_i}
    \ket{\psi_i}=\otimes_{j=1}^n \left(\sum_{e_j\in\bbZ_{k_j}}\ket{x_{e_{j}}^{(j)}}\right)_{A_j}\in \cA_i.
\end{equation}
The ``stopper" state is
\begin{equation}\label{eq:genS}
  \ket{S}=\otimes_{j=1}^n\left(\sum_{r\in\bbZ_{d_j}}\ket{r}\right)_{A_j}.
\end{equation}
We may wonder whether the OPS 
\begin{equation}\label{eq:X}
   \cX:=\cup_{i=1}(\cA_i\setminus\{\ket{\psi_i}\})\cup\{\ket{S}\} 
\end{equation} is a UPB that is still a UPB in every bipartition. 

For any bipartition $C\mid D$, where $C,D\subset \{1,2,\ldots ,n\}$, and $C\cup D= \{1,2,\ldots ,n\}$, the OPB $\cB$ must correspond to a  tile structure $\cT_{C\mid D}$ in $\bbC^{h_1}\otimes \bbC^{h_2}$, where $h_1=\prod_{g\in C}d_{g}$ and  $h_2=\prod_{g\in D}d_{g}$ (For example, see Fig.~\ref{fig:tile39}). By using Lemma~\ref{lem:upb}, if any $r$ ($2\leq r\leq s-1$) tiles in $\cT_{C\mid D}$ cannot form a rectangle, then $\cX_{C\mid D}$ is a UPB in $\bbC^{h_1}\otimes \bbC^{h_2}$. Note that in this case, $\cX$ is also a UPB in $\bbC^{d_1}\otimes \bbC^{d_2}\otimes \cdots \otimes \bbC^{d_n}$. This is because if $\cX$ is not a UPB, then there exists a product state $\ket{\psi}$ in $\cX^{\bot}$, and the product state $\ket{\psi}_{C\mid D}$ in bipartition $C|D$   belongs to $\cX_{C\mid D}^{\bot}$, which contradicts $\cX_{C\mid D}$ being a UPB in $\bbC^{h_1}\otimes \bbC^{h_2}$. Now, we have the following theorem.
\vspace{0.3cm}
\begin{theorem}\label{thm:gupb}
Consider a tile structure $\cT=\cup_{i=1}^s t_i$ ($s\geq 5$) in $\bbC^{d_1}\otimes \bbC^{d_2}\otimes \cdots \otimes \bbC^{d_n}$. For any bipartition $C\mid D$, if  any $r$ ($2\leq r\leq s-1$) tiles in $\cT_{C\mid D}$ cannot form a rectangle, then the OPS $\cX$ given by Eq.~\eqref{eq:X} is a UPB in $\bbC^{d_1}\otimes \bbC^{d_2}\otimes \cdots \otimes \bbC^{d_n}$, which is still a UPB in every bipartition.
\end{theorem}

One can use computer to search the tile structure in Theorem~\ref{thm:gupb}. By exhaustive search, we show that such a tile structure does not exist in $\bbC^3\otimes \bbC^3\otimes \bbC^3$. However, we conjecture that such a tile structure may exist in a higher multipartite system.

\section{conclusion and discussion}\label{sec:con}
In this paper, we showed that there exist some  unextendible product bases that are   uncompletable product bases in every bipartition in $\bbC^{d_1}\otimes \bbC^{d_2}\otimes \bbC^{d_3}$ and $\bbC^{d_1}\otimes \bbC^{d_2}\otimes \bbC^{d_3}\otimes \bbC^{d_4}$ for $d_1,d_2,d_3,d_4\geq 3$, and $\bbC^3\otimes \bbC^3\otimes \bbC^3$ achieved the minimum system for the existence of such unextendible product bases. This result answers an open question proposed in \cite{divincenzo2003unextendible}.   We also showed that such unextendible product bases can be used for local hiding of information. Finding an unextendible product basis that is still an unextendible product basis in every bipartition can be challenging, and we gave a sufficient condition for the existence of such  an
unextendible product basis.

There are some interesting open questions left. How to find an unextendible product basis that is still an unextendible product basis in every bipartition by using Theorem~\ref{thm:gupb}? What is the minimum size of unextendible product basis that is an uncompletable product basis in every bipartition?

\section*{Acknowledgments}
\label{sec:ack}		
We thank Lin Chen for discussing this problem. The research of X.Z. and F.S. were supported by the NSFC under Grants No. 12171452 and No. 11771419,
the Anhui Initiative in Quantum Information Technologies under Grant No. AHY150200, the National Key
Research and Development Program of China (2020YFA0713100), and the Innovation Program for Quantum
Science and Technology (2021ZD0302904). M.-S.L. was supported by the NSFC under Grant No. 12005092, the China Postdoctoral Science
Foundation (2020M681996).

	\bibliographystyle{IEEEtran}
	\bibliography{reference}	

\begin{thebibliography}{10}
\providecommand{\url}[1]{#1}
\csname url@samestyle\endcsname
\providecommand{\newblock}{\relax}
\providecommand{\bibinfo}[2]{#2}
\providecommand{\BIBentrySTDinterwordspacing}{\spaceskip=0pt\relax}
\providecommand{\BIBentryALTinterwordstretchfactor}{4}
\providecommand{\BIBentryALTinterwordspacing}{\spaceskip=\fontdimen2\font plus
\BIBentryALTinterwordstretchfactor\fontdimen3\font minus
  \fontdimen4\font\relax}
\providecommand{\BIBforeignlanguage}[2]{{%
\expandafter\ifx\csname l@#1\endcsname\relax
\typeout{** WARNING: IEEEtran.bst: No hyphenation pattern has been}%
\typeout{** loaded for the language `#1'. Using the pattern for}%
\typeout{** the default language instead.}%
\else
\language=\csname l@#1\endcsname
\fi
#2}}
\providecommand{\BIBdecl}{\relax}
\BIBdecl

\bibitem{bennett1999unextendible}
C.~H. Bennett, D.~P. DiVincenzo, T.~Mor, P.~W. Shor, J.~A. Smolin, and B.~M.
  Terhal, ``Unextendible product bases and bound entanglement,'' \emph{Physical
  Review Letters}, vol.~82, no.~26, p. 5385, 1999.

\bibitem{bennett1999quantum}
C.~H. Bennett, D.~P. DiVincenzo, C.~A. Fuchs, T.~Mor, E.~Rains, P.~W. Shor,
  J.~A. Smolin, and W.~K. Wootters, ``Quantum nonlocality without
  entanglement,'' \emph{Physical Review A}, vol.~59, no.~2, p. 1070, 1999.

\bibitem{cohen2008understanding}
S.~M. Cohen, ``Understanding entanglement as resource: Locally distinguishing
  unextendible product bases,'' \emph{Physical Review A}, vol.~77, no.~1, p.
  012304, 2008.

\bibitem{zhang2020locally}
Z.-C. Zhang, X.~Wu, and X.~Zhang, ``Locally distinguishing unextendible product
  bases by using entanglement efficiently,'' \emph{Physical Review A}, vol.
  101, no.~2, p. 022306, 2020.

\bibitem{Halder2019Strong}
S.~Halder, M.~Banik, S.~Agrawal, and S.~Bandyopadhyay, ``Strong quantum
  nonlocality without entanglement,'' \emph{Physical Review Letters}, vol. 122,
  no.~4, p. 040403, 2019.

\bibitem{shi2021strong}
F.~Shi, M.-S. Li, L.~Chen, and X.~Zhang, ``Strong quantum nonlocality for
  unextendible product bases in heterogeneous systems,'' \emph{Journal of
  Physics A: Mathematical and Theoretical}, vol.~55, no.~1, p. 015305, 2021.

\bibitem{shi2022strongly}
F.~Shi, M.-S. Li, M.~Hu, L.~Chen, M.-H. Yung, Y.-L. Wang, and X.~Zhang,
  ``Strongly nonlocal unextendible product bases do exist,'' \emph{Quantum},
  vol.~6, p. 619, 2022.

\bibitem{bandyopadhyay2011more}
S.~Bandyopadhyay, ``More nonlocality with less purity,'' \emph{Physical Review
  Letters}, vol. 106, no.~21, p. 210402, 2011.

\bibitem{augusiak2011bell}
R.~Augusiak, J.~Stasi{\'n}ska, C.~Hadley, J.~Korbicz, M.~Lewenstein, and
  A.~Acin, ``Bell inequalities with no quantum violation and unextendable
  product bases,'' \emph{Physical Review Letters}, vol. 107, no.~7, p. 070401,
  2011.

\bibitem{Augusiak2012tight}
R.~Augusiak, T.~Fritz, M.~Kotowski, M.~Kotowski, M.~Pawlowski, M.~Lewenstein,
  and A.~Acin, ``Tight {B}ell inequalities with no quantum violation from qubit
  unextendible product bases,'' \emph{Physical Review A}, vol.~85, no.~4, p.
  042113, 2012.

\bibitem{divincenzo2003unextendible}
D.~P. DiVincenzo, T.~Mor, P.~W. Shor, J.~A. Smolin, and B.~M. Terhal,
  ``Unextendible product bases, uncompletable product bases and bound
  entanglement,'' \emph{Communications in Mathematical Physics}, vol. 238,
  no.~3, pp. 379--410, 2003.

\bibitem{demianowicz2018unextendible}
M.~Demianowicz and R.~Augusiak, ``From unextendible product bases to genuinely
  entangled subspaces,'' \emph{Physical Review A}, vol.~98, no.~1, p. 012313,
  2018.

\bibitem{demianowicz2022genuinely}
M.~Demianowicz, ``Genuinely entangled subspaces of maximal dimensions (and many
  smaller ones) cannot be constructed from orthogonal unextendible product
  bases,'' \emph{arXiv preprint arXiv:2202.08356}, 2022.

\bibitem{shi2020unextendible}
F.~Shi, X.~Zhang, and L.~Chen, ``Unextendible product bases from tile
  structures and their local entanglement-assisted distinguishability,''
  \emph{Physical Review A}, vol. 101, no.~6, p. 062329, 2020.

\bibitem{horodecki1997separability}
P.~Horodecki, ``Separability criterion and inseparable mixed states with
  positive partial transposition,'' \emph{Physics Letters A}, vol. 232, no.~5,
  pp. 333--339, 1997.

\bibitem{halder2019family}
S.~Halder, M.~Banik, and S.~Ghosh, ``Family of bound entangled states on the
  boundary of the peres set,'' \emph{Physical Review A}, vol.~99, no.~6, p.
  062329, 2019.

\bibitem{agrawal2019genuinely}
S.~Agrawal, S.~Halder, and M.~Banik, ``Genuinely entangled subspace with
  all-encompassing distillable entanglement across every bipartition,''
  \emph{Physical Review A}, vol.~99, no.~3, p. 032335, 2019.

\bibitem{shi2022strong}
F.~Shi, Z.~Ye, L.~Chen, and X.~Zhang, ``Strong quantum nonlocality in
  {$N$}-partite systems,'' \emph{Physical Review A}, vol. 105, no.~2, p.
  022209, 2022.

\end{thebibliography}
	
\end{document}